\documentclass[conference,10pt]{IEEEtran}
\makeatletter
\def\ps@headings{%
\def\@oddhead{\mbox{}\scriptsize\rightmark \hfil \thepage}%
\def\@evenhead{\scriptsize\thepage \hfil \leftmark\mbox{}}%
\def\@oddfoot{}%
\def\@evenfoot{}}
\makeatother
\usepackage{latexsym}
\usepackage{amsmath}
\usepackage{amssymb}
\usepackage{amsthm}
\usepackage{epsfig}
\usepackage{algorithm}
\usepackage{algorithmic}
\usepackage{color}
\usepackage{tikz}
\usepackage{mathdots}

\usepackage{multicol}

\usetikzlibrary{calc}
\usetikzlibrary{shapes}
\usetikzlibrary{shadows}

\newtheorem{theorem}{Theorem}

\newtheorem{fact}{Fact}

\def\N{\mathbb{N}}

\def\R{\mathbb{R}}

\def\O{\mathcal{O}}
\def\1{\mathbf{1}}

\def\NP{\mathbf{NP}}

\def\rnd{\textsc{RND}}

\def\sp{\mathrm{SP}}
\def\tr{\mathrm{TR}}
\def\hub{\mathrm{HUB}}
\def\hh{\mathrm{HH}}
\def\vpn{\mathrm{VPN}}

\def\srlb{\mathrm{SRLB}}

\def\rndsp{\rnd_{\sp}}
\def\rndtr{\rnd_{\tr}}
\def\rndhub{\rnd_{\hub}}
\def\rndhh{\rnd_{\hh}}

\def\rndx{\rnd_{X}}

\def\trunc{\mathrm{trunc}}

\def\dist{\mathrm{dist}}

\def\conv{\mathrm{conv}}

\usepackage{fmtcount}
\begin{document}

\title{Shortest Path versus Multi-Hub  Routing in Networks with Uncertain Demand}
\author{
\IEEEauthorblockN{
Alexandre Fr\'echette\IEEEauthorrefmark{1},
F. Bruce Shepherd\IEEEauthorrefmark{2},
Marina K. Thottan\IEEEauthorrefmark{3} and
Peter J. Winzer\IEEEauthorrefmark{4}
}
\IEEEauthorblockA{\IEEEauthorrefmark{1}Department of Computer Science, University of British Columbia, Canada$^{1}$}
\IEEEauthorblockA{\IEEEauthorrefmark{2}Department of Mathematics and Statistics, McGill University, Canada}
\IEEEauthorblockA{\IEEEauthorrefmark{3}Bell Labs, Emerging Technology Division, Murray Hill, NJ, USA}
\IEEEauthorblockA{\IEEEauthorrefmark{4}Bell Labs, Optical Transmission and Networks Research Department, Holmdel, NJ, USA}
}

\maketitle

\addtocounter{footnote}{1}
\footnotetext[\value{footnote}]{This author's work was done while studying at McGill University, and with generous funding from NSERC Discovery Grant No. 212567.}

\begin{abstract}
We study a class of robust network design problems motivated by the
need to scale core networks to meet increasingly dynamic capacity
demands. Past work has focused on designing the network to support
all hose matrices (all matrices not exceeding marginal bounds at the
nodes). This model may be too conservative if additional
information on traffic patterns is available. Another extreme is the
fixed demand model, where one designs the network to support
peak point-to-point demands. We introduce a capped hose model to
explore a broader range of traffic matrices which includes the above two as special cases.
 It is known that optimal designs for the hose model are always determined by
single-hub routing, and for the fixed-demand model are based on
shortest-path routing.
We shed light on the wider space of capped
hose matrices in order to see which traffic models are more
shortest path-like as opposed to hub-like. To address the space in between,
we use hierarchical multi-hub routing templates, a generalization of hub and tree routing. In particular, we show that by adding peak capacities into the hose model, the single-hub tree-routing template is no longer cost-effective.
This initiates the study of a class of robust network design ($\rnd$) problems restricted to these templates.
Our empirical analysis
is based on a  heuristic for this new hierarchical $\rnd$ problem. We also propose that it is possible to define a routing indicator that accounts for the strengths of the marginals and peak demands and use this information to choose the appropriate routing template.
We benchmark our approach against other well-known routing
templates, using representative carrier networks and a variety of different capped hose traffic demands, parameterized by the relative importance of their marginals as opposed to their point-to-point peak demands. This study also reveals conditions under which multi-hub routing gives improvements over single-hub and shortest-path routings.
\end{abstract}

\section{Introduction}
Traditional communication networks are designed based on  knowledge of an {\em expected traffic demand matrix} that specifies the aggregate traffic between every pair of nodes in the network and evolves slowly, over a time scale of months or years. In an era of data-dominant communication networks with dynamic demand patterns, this traditionnal network design methodology loses its cost-effectiveness. For instance, in emerging data services such as YouTube, user generated video content and flash crowds are triggering an increasing amount of dynamic uncertainty in the traffic demand. One specific concern is that the real-time estimation of the dynamically changing traffic patterns in large data networks is intractable. Moreover, a pessimistic approach of designing for peak point-to-point demands (also referred to as the fixed-demand model) throws away critical information for achieving ``cost-sharing'' between the varying demand patterns that arise over time.

A traffic model that has gained considerable popularity in coping with dynamic demands is the {\em hose model} \cite{duffield2002resource,fingerhut1997designing}.
In this model, only bounds  on the ingress/egress traffic of the nodes are known.  These bounds (called {\em marginals}) are physically represented by the various interface speeds of networking hardware at any given node.
The actually encountered point-to-point traffic distributions are subject to dynamic variations during network operation, but the network is provisioned to be able to support any possible pattern meeting the ingress/egress bounds. Formally, in the symmetric (undirected) demand case, the space of hose matrices $\mathcal{H}_U$ consists of symmetric
matrices $D$ satisfying the bound: $\sum_j D_{ij} \leq U(i)$ for each node $i$, with marginal $U(i)$; in addition, the main diagonal of these matrices is typically set to zero, reflecting the fact that nodes generally do not send traffic to themselves. However, the flexibility of the hose model can also be a drawback since totally random demand matrices rarely ever occur in practice, and the operator usually has additional at least some a priori information about the traffic pattern. For instance, it may happen that there exists some additional peak demand $U(i,j)$ between nodes $i,j$, i.e. $D_{ij} \leq U(i,j)$. We call this union of marginals and peak demands the {\em capped hose model}. Depending on the choice of $U(i,j)$, the space of traffic patterns spanned by the capped hose model ranges from purely deterministic point-to-point demands ($D_{ij}=U(i,j)$) to the space of unconstrained hose traffic ($U(i,j)=\min(U(i),U(j))$). Designing a network for the capped hose model hence avoids the massive over-provisioning that may arise for unconstrained hose traffic \cite{shepherd2006selective}.

In order to solve the routing problem for the capped hose model, we propose the use of a class of oblivious routing strategies, called {\em hierarchical hubbing routing templates} ($\hh$), which generalizes the popular hub and tree routing templates.
We give empirical results based on a heuristic for the corresponding new class of robust network design ($\rnd$) problems.

In order to quantify the extent of the space of traffic matrices spanned by the capped hose model,
 we propose indicators (marginal versus peak strength measures) that could be used as a predictor for which traffic scenarios favour a shortest path routing ($\sp$), and which favour $\hh$ routings.
We also show that $\hh$ is often superior in general to both $\sp$ and hub routing ($\hub$), in terms of network design (considering both link and node costs).

\subsection{Summary of Main Contributions}
\begin{enumerate}
\item We systematically study the capped hose model as a generalization of the fixed-demand and the unconstrained hose models. This allows tapping into a whole range of practically relevant traffic scenarios.
\item We initiate a formal study of a class of $\rnd$ problems based on oblivious $\hh$ routing templates, a generalization of tree and single-hub routing used by optimal virtual private networks ($\vpn$s). 
We provide theoretical examples showing that $\sp$ is much cheaper than $\hh$, and vice-versa.
\item We perform an empirical study (based on both randomly generated and real traffic data) on the space of new capped hose traffic polytopes.
\item Our initial empirical findings show that for most instances, $\hh$ is better than single-hub $\vpn$, suggesting that multi-hubbing at different layers can be important for network cost savings. Empirical results suggest that for many instances $\hh$ is significantly better than $\sp$.
\item We provide some initial understanding of which traffic scenarios favour $\sp$-like, and which are more $\hh$-like. This is achieved by introducing two metrics that quantify the strengths of the marginals and peak demands in the specification of a collection of traffic matrices. These indicators could be used as a predictor for which of the two routing templates, $\sp$ and $\hh$, is better.
\end{enumerate}

\section{Traffic Engineering for Core Networks}

\subsection{Oblivious Routing and Routing Templates}

Apart from the {\em traffic modelling} complexities mentioned above, which form the basis for network infrastructure {\em planning} and {\em design}, there is
a second {\em operational} complication arising in any network with dynamically changing traffic patterns: how should traffic be routed on an installed infrastructure? Since centralized or distributed control planes may not have the speed and flexibility to adapt dynamically to rapidly varying traffic patterns without the risk of inducing detrimental oscillations \cite{gao2006avoiding},
the vast majority of today's networks are designed based on the
principle  of {\em oblivious
routing}, usually based on $\sp$, both on the
circuit and on the packet layer.

Informally,
Valiant \cite{valiant1982scheme} defined oblivious routing as ``the route taken by each packet be determined entirely by itself. The other packets can only
influence the rate at which the route is traversed.'' Formally, this amounts to specifying a {\em template}:
for each possible communicating pair of nodes $i,j$, one specifies a path $P_{ij}$.\footnote{Valiant actually considered the randomized version and so the flows could be fractional.}
Traffic from node $i$ to $j$ is always routed on $P_{ij}$ independent of any other demands or
congestion in the network.
Oblivious routing is attractive because it enables strictly local routing
decisions. In contrast, stable non-oblivious routing would require a global network control plane performing path
optimization on a time scale much faster than the congestion and traffic variation dynamics, which is impractical in most scenarios.

For the {\em fixed-demand traffic model}, a minimum cost design  is
achieved using an oblivious $\sp$  routing template.
 Using $\sp$ for hose traffic, however, may result in significant over-provisioning of
network resources (if blocking is to be avoided)\cite{shepherd2006selective}. On the other hand,
it was proved in \cite{goyal2008vpn}
that  the optimal design (in the undirected setting) for the {\em hose model} is always
induced by  a tree routing template  (i.e. there exists a network-specific tree, and the template is to use the path $P_{ij}$ in that tree between node $i,j$). In  \cite{fingerhut1997designing,gupta2001provisioning}
a simple polynomial-time algorithm was given to produce the optimal tree template.
Their analysis shows that  the resulting network (called a virtual private network - $\vpn$) actually has enough capacity to support
{\em hub routing} ($\hub$) within the tree. That is, there is a hub node $h$ and enough capacity in the tree for  each node to reserve a private  circuit
of size $U(i)$ to $h$. We use $\hub$ to denote such routing templates.

\subsection{The Capped Hose Model}
\label{sec:cappedhose}
We now consider the class of { \em capped hose matrices } and the minimum cost network to support it.
\begin{itemize}
\item If the node bounds
$U(i)$ are sufficiently large (at least $\sum_j U(i,j)$),   then they impose no constraints on the routed demand.
Hence the network is only required to
support the single fixed-demand matrix $U(i,j)$.  In this case, we know that
an optimal design  results from $\sp$ routing.
\item If the point-to-point $U(i,j)$ are large (at least $\min\{U(i),U(j)\}$), then they impose no constraint. Hence the   network must support all hose matrices associated
with the marginals $U(i)$. In this case, we know that $\hub$ routing is optimal.
\item In the general case where neither marginals nor peak demands clearly dominate the space of possible traffic matrices, we see that there are gaps at the two ends of the spectrum: at the fixed-demand end, we may  incur large penalties if we design the network via hub routing. Similarly, at the unconstrained hose-demand end we may incur large penalties if we use $\sp$ (see Section~\ref{sec:gaps}).
\end{itemize}
One would normally assume that a node's marginal is big enough to allow the routing of any of its peak demands. Hence, we call $[\max_j U(i,j),\sum_j U(i,j)]$ the {\em relevance interval} of a marginal $U(i)$ for a given set of peak demands. A similar reasoning for peak demands suggests that $U(i,j)$'s relevance interval is $[0,\min(U(i),U(j))].$

In order to understand the spectrum of traffic matrices between these  extremes it is necessary to consider more general oblivious routing strategies. To solve this problem, we propose the use of $\hh$ (see Section \ref{sec:HH}).  On the other hand, to characterize when we are close to the extremes and can rely on standard $\hub$ and $\sp$ routings, we introduce a pair of metrics.
\subsection{Marginal and Peak Demand Strength Metrics}
\label{sec:strength}
For any capped hose model, we propose two vector measures $\pi$ (for peak) and $\mu$ (for marginal) to indicate the relevance of the different traffic measurements
in defining a given collection of traffic  patterns. For example, if $\mu(i)$ is large, it means that the marginal $U(i)$ should play a significant role. On the other hand if $\pi(i,j)$ is large, the peak demand $U(i,j)$ has a more significant role in determining a cost efficient routing template. One aspect of this study is to examine whether these $\pi$ and $\mu$ measures can be used to classify traffic instances as being more favourable for a $\sp$ or  hub-like routing template.


Given marginals and peak demands, we let the \emph{strength} of node $i$'s marginal be:
$$
\mu(i) = 1-\frac{\trunc(U(i)) -\max_j U(i,j)}{\sum_j U(i,j) - \max_j U(i,j) }.
$$
where $\trunc(U(i)) = \min(U(i), \sum_j U(i,j))$.
 Similarly, the \emph{strength} of the peak demand between two nodes $i,j$ is:
$$
\pi(i,j) = 1-\frac{\trunc(U(i,j))}{\min(U(i),U(j)}.
$$
where $\trunc(U(i,j)) = \min(U(i,j), \min(U(i),U(j)))$. These metrics are faithful to our observations  and range from $0$ (weak) to $1$ (strong).

We thus define our strength vectors $\mu\in \R^n$ for marginals and $\pi \in \R^{n\times n}$ for peak demands, with $\mu_i = \mu(i)$ and $\pi_{i,j} = \pi(i,j)$ respectively. A capped hose instance that is hose-like should have relatively high marginal strengths, and if it is peak-demand like it should have high peak demand strengths if it is peak-demand-like. Based on this intuition, we attempt to  classify a capped hose instance by the Euclidean norm of these vectors.

\subsection{Hierarchical Hub Routing}

It is well-known that designing a network for hose matrices is quantitatively similar (e.g., up to a factor 2 in capacity) to designing for a single ``uniform multiflow'' instance. Since uniform multiflow instances can be viewed as a convex combination of hub routings (\cite{valiant1982scheme,shepherd2006selective}),  it seems natural that one of these hubs  would deliver  a good overall network design in terms of link capacity cost. However, in the capped hose model, we may ``punch holes'' into the space of demand matrices by letting some $U(i,j)$ be much smaller than others. This allows for the existence of certain regions within which there is much denser traffic than between such regions.
Since these regions  have their own capacity sharing benefits, one
should no longer expect a single hub, but multiple hubs, each serving its own region of dense traffic.

This forms one motivation for our choice  of {\em  hierarchichal hub routing templates} ($\hh$).
Each $\hh$ template is induced by an associated {\em hub tree}. If the hub tree is a star, it corresponds to standard hub routing. If the hub tree has more layers, its internal   nodes represent possible hub nodes for different subnetworks (these subnets are necessarily nested due to the tree structure).
A detailed description is given in Section \ref{sec:HH}.

\subsection{The Robustness Paradigm}
We are given an undirected network topology with per-unit costs of reserving bandwidth. In addition we are given a space of demand matrices (in our case, the  capped hose model) which the network must support
via oblivious routing. This space, or ``traffic model'', is meant to capture the time series of all possible
demand matrices which the network may encounter. The goal is to find a routing template
which minimizes the overall cost of reserved capacity needed to support all demands in the traffic space. The formal definition of robust network design (RND) appears in Section~\ref{sec:RNDmodel}.

Due to the optimality arguments of various routing templates for the traffic models under consideration, we focus on RND restricted to the routing templates $\sp$, $\hub$, $\tr$ (tree routing) and $\hh$. For
a class of templates $X$, we denote by $\rndx$ the  RND problem where one must restrict to the
corresponding class of templates. While $\rndsp$, $\rndtr$ and $\rndhub$
 have been
well-studied, we initiate a formal analysis of $\rndhh$ in Section~\ref{sec:HH}.
We also propose a heuristic for this version in Section~\ref{sec:heuristic}.

\subsection{Related Work}
Oblivious routing approaches to network optimization have been used
in many different contexts from switching (\cite{mitra1987randomized,widjaja2003exploiting,keslassy2003scaling})
to overlay networks
(\cite{kodialam2004efficient,zhang2004designing}), or fundamental
tradeoffs in distributed computing \cite{valiant1982scheme,racke2002minimizing} to
name a few. In each case mentioned above, the primary performance
measure is network congestion (or its dual problem throughput). In
fact, one could summarize the early work by Valiant, Borodin and
Hopcroft as saying that randomization is necessary and sufficient
for oblivious routing  to give $\O(\log n)$ congestion in
many packet network topologies; $n$ being the number of nodes.

The present work's focus is not on congestion, but on total link capacity cost.
This falls in the general space of $\rnd$ problems, which includes
the Steiner tree and VPN problems as special cases.
Exact or constant-approximation polytime algorithms for several other
important traffic models   have also been designed, e.g. the so-called symmetric and
asymmetric hose models \cite{gupta2007approximation,olver10approx}.
The latter is however $\NP$-hard to approximate in undirected graphs within
polylogarithmic factors for general polytopes \cite{olver10approx}.

Our work was partly motivated by  work in  \cite{shepherd2006selective}, where
Selective Randomized Load Balancing ($\srlb$) is used to
design minimum-{\em cost} networks. They showed that networks whose design is based on oblivious
routing techniques (and specifically $\srlb$)
can be
ideal to capture cost savings in IP networks. This is  because optimal designs reserve capacity
on long paths, where one may  employ
high-capacity optical circuits, and partially avoid expensive IP equipment costs at internal nodes
(their empirical study  incorporated  IP router, optical switching, and fiber costs).
In this paper, instead of a design based on
hub  routing to a small number of hubs, we consider more
general $\hh$, as discussed in \cite{olver10approx}.

Our work also extends the long stream of work on designing $\vpn$s \cite{duffield2002resource,gupta2001provisioning,goyal2008vpn,grandoni07vpn,ghobadi2008resource}.
Previous work has focused on provisioning for a $\vpn$ based on either the hose
or the fixed-demand models (c.f. \cite{kumar2001algorithms}). Designing for more general traffic
models in this context has not received  much attention. We take an intermediate step by examining
the capped hose model
which to the best of our knowledge has not been studied.

With some thought, the new class of $\rndhh$ problems is almost equivalent
to solving $\rnd$ but only using  tree routing in the metric completion of the graph.
The latter approach is exactly the idea used by Anupam Gupta (unpublished) to show an $\O(\log n)$-approximation
for general $\rnd$ via metric tree embeddings (cf. a long version of \cite{goyal2009dynamic} contains full details).

\section{The Robust Network Design (RND) Model}
\label{sec:RNDmodel}
To understand the hierarchical hub routing problem we first introduce the robust network design
 paradigm in full generality.

\subsection{Formal Definition}
\label{sec:rndoverview}
An instance of $\rnd$ consists of a {\em
network topology} (undirected graph) $G=(V,E)$, as well as per-unit
costs $c(e)$ of bandwidth on each link $e \in E$ (the edge costs may be implicit if $G$ is weighted).
In addition, we are given a convex region $\mathcal{D} \subseteq \mathbb{R}^{V \times
V}$.\footnote{We always assume that $\mathcal{D}$ is itself well-described in
the sense that we can solve its separation problem efficiently, e.g.
in polytime.} The region $\mathcal{D}$ represents the {\em traffic model} or  {\em demand universe}, that is, the set of demands
which have to be supported. In other words,  any matrix
$D \in \mathcal{D}$ must be routable in the capacity we buy on $G$. The
cost of the network increases as this space grows.

One may define  several versions of $\rnd$ depending on assumptions
on the routing model. Most common in practice (and the focus of this paper) is the {\em oblivious
routing} model described above. These paths often need to satisfy additional requirements, such as be shortest paths in $\sp$ or form a tree in $\tr$. This is specified by a class of template $X$.

Our main objective in the context of $\rnd$
is then to compute an oblivious routing template in a specific class of templates which leads to a link
capacity vector $\mathsf{cap}_e$ whose overall cost is minimized (as we discuss
next). We use $\rndx(G,c,\mathcal{D})$ to denote the cost of an optimal solution

Given  a fixed routing template $\mathcal{P}=(P_{ij}: \forall ~i,j)$, the
cost of supporting all demands  in the demand polytope $P$ can be
computed in a straightforward fashion. For each edge $e \in E(G)$,
we compute a ``worst demand'' matrix in $P$ in terms of routing on
$e$ (using the routing $\mathcal{P}$). This gives rise to the
following sub-optimization:
\begin{equation}
\mathsf{cap}_e (\mathcal{P},\mathcal{D}) := \max_{D \in \mathcal{D}} ~\sum_{\substack{i,j \in V: \\ e \in P_{ij}}} D_{ij}.
\end{equation}

\noindent $\rnd$ then asks for a template $\mathcal{P}$ that minimizes the cost of the overall required capacity:
$$
\rndx(G,c,\mathcal{D}) = \min_{\mathcal{P} \in X} ~\sum_{e\in E} c(e) \mathsf{cap}_e(\mathcal{P},\mathcal{D}) .
$$

\subsection{The Hierarchical Hub Routing Template ($\hh$)}
\label{sec:specialrouting}

We make heavy use of the generalization  $\hh$ \cite{olver10approx}.
Such oblivious routing templates are derived from a tree $T$ whose leaves
consist of the terminal nodes, that is $V(G)$ without loss of generality. We call this the {\em hub tree}.
The internal nodes are ``abstract'' in the sense that they represent nodes (yet to be identified in $G$)
which will act as hubs for certain subsets of terminal nodes. For instance,
for any non-leaf node $v \in T$, let $T_v$ be the subtree rooted at $v$ dangling below it.
The interpretation is that the leaves $L_v$ in $T_v$ represent a region of terminal nodes which should have
a dedicated hub (this hub will correspond to $v$).
The family of regions arising from a tree obviously forms a nested
 family (Figure 1).

Obtaining an oblivious routing template from the $\hh$ class is a two step process. First, we first need to pick a hub tree, and then map its internal node to physical nodes. In this second stage, we find a map from each internal node $v$ to a real node in $G$ to act as a hub. Such
 a \emph{hub map} $\eta:V(T) \rightarrow V(G)$ must satisfy $\eta(v)=v$ for each leaf of $T$.
Given such a mapping, there is an {\em induced} oblivious routing template  defined as follows:
for each pair of leaf nodes $i,j$ in $T$, let $i=v_0,v_1,v_2, \ldots ,v_l=j$ be the unique $ij$ path in $T$.
For each $q=0,1, \ldots, l-1$, let $P_q$ denote a shortest $\eta(v_q)\eta(v_{q+1})$ path in $G$. Then
we define $P_{ij}$ to be the concatenation of paths $P_0,P_1, \ldots ,P_{l-1}$
(resulting in a possibly non-simple path). Naturally, the special case of $\hub$ routings corresponds to taking $T$ equal to a star.
\begin{figure}
\begin{center}
\begin{tikzpicture}[scale=0.35]
\node[inner sep=1pt,draw,circle](v1) at (0,0){};
\node[inner sep=1pt,draw,circle](v2) at (1,0){};
\node[inner sep=1pt,draw,circle](v3) at (0.5,-0.5){};

\node[inner sep=1pt,draw,circle](v4) at (4,1){};
\node[inner sep=1pt,draw,circle](v5) at (4.5,1.5){};

\node[inner sep=1pt,draw,circle](v6) at (-2,3){};
\node[inner sep=1pt,draw,circle](v7) at (-2.5,2){};

\draw[dashed] (0.5,-0.25) circle (1);

\draw[dashed] (4.25,1.25) circle (1);

\draw[dashed] (-2.25,2.37) circle (1);

\draw[dashed] (-1,1) circle (3.5);

\draw[dashed] (0,1) circle (5.5);

\node[inner sep=1pt,draw,circle,fill=black](h1) at (0.5,0){};
\draw (h1)--(v1);
\draw (h1)--(v2);
\draw (h1)--(v3);

\node[inner sep=1pt,draw,circle,fill=black](h2) at (-1.6,2.35){};
\draw (h2)--(v6);
\draw (h2)--(v7);

\node[inner sep=1pt,draw,circle,fill=black](h3) at (-0.55,1.17){};
\draw (h3)--(h1);
\draw (h3)--(h2);

\node[inner sep=1pt,draw,circle,fill=black](h4) at (4,1.75){};
\draw (h4)--(v4);
\draw (h4)--(v5);

\node[inner sep=1pt,draw,circle,fill=black](h5) at (1.75,5){};
\draw (h5)--(h4);
\draw (h5)--(h3);

\end{tikzpicture}
\end{center}
\label{fig:Tree}
\caption{A hub tree. A ``region'' of terminals is represented by a dashed circle. Internal black nodes of the tree are the ones that will be mapped to hubs in the network topology. White leaf nodes are terminals.}
\end{figure}
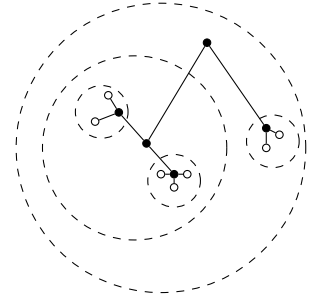
Each hub tree thus  gives
rise to many possible routing templates, one for each mapping $\eta$ (analogous to  the specific choice of hub node in hub routing).
Call this the class of hierarchical hub templates.
Note that, in some scenarios, there are single-path oblivious routing templates that are unachievable by hierarchical hub templates.
For instance in a complete graph on four nodes, with unit weight edges, there is no
hierarchical hub template which corresponds to the shortest path template (each $P_{ij}$ consist of the edge $ij$).

$\hh$ routing also responds to
one practical impediment of hub routing architectures, pointed out in \cite{shepherd2006selective}. Namely,  network providers are generally opposed to having all traffic routed via a single
hub (or a small cluster of hubs) in the center of a large network. In particular, traffic within some local
region should be handled by hubs within their own ``jurisdiction".
Hub trees potentially give a mechanism for specifying these local jurisdictions.

\section{The Optimization Problem $\rndhh$}
\label{sec:HH}

There are several  layers to the optimization problem $\rndhh$. First, we must find a hub tree $T$ - this is the core decision for which we provide a heuristic algorithm in Section \ref{sec:heuristic}. Second, given a hub tree, we must map it to the network graph to obtain a valid routing template. An efficient solution to this mapping problem is known for $T$-topes\footnote{A $T$-tope is a polytope of demands, denoted by $\mathcal{H}_{T,b}$, which are routable on some tree $T$ with edge capacities $b$. This class generalizes hose polytopes, which can be viewed as demands routable on a star.} \cite{olver10approx}; we show that this can be invoked to solve the  problem for general polytopes in Section \ref{sec:hubplacement}. Finally, given a template, we must ultimately solve for capacities it induces, as we discuss in Equation (\ref{eqn:caps}) below.

We  mention that the case where the hub tree is given is of interest
in cases where a  network provider wishes to self-identify a regional clustering
within their network. In other words, they provide the hub tree, and it only   remains to
identify which nodes to act as hubs (i.e. solve the hub placement problem).

In Section \ref{sec:cases} we  discuss several classical  problems which arise as special cases of  $\rndhh$.  Finally, we provide some theoretical bounds between $\rndhh$ and $\rndsp$ in Section~\ref{sec:gaps}.

\subsection{Choosing the Hub Tree}
\label{sec:heuristic}
We now develop  a heuristic for $\rndhh$. It is tailored for the capped hose model,
but could be extended to handle  any traffic polytope. The high-level intuition is to
group nodes so that  traffic  within  a group is large compared to  traffic leaving the group. This
 suggests the need for a hub to handle the group's local traffic.
It is straightforward to show that there is always  an optimal hub tree for $\rndhh$ where the tree is binary (we defer the details to the end of this section).
Hence our algorithm's focus is to find a ``good'' binary hub tree.

To construct a binary tree from the terminal nodes $V$, one could either proceed by a top-down approach (recursively splitting groups of nodes in two) or a bottom-up approach (recursively merging nodes together). The former has the
flavour of repeatedly solving   \textsc{Sparsest Cut}  problems. This problem is $\mathbf{APX}$-hard
and the current best approximation factors are polylogarithmic
 \cite{chawla06NPhard,arora2009expander}. Instead,  we  follow the bottom-up approach. We are thus seeking
to repeatedly merge pairs of node sets which share a lot of traffic with each other.
The key to this calculation is a sparsest cut-like measure which we introduce next.
This is then used in a
 \textit{Binary Tree Sparsest Merging (BTSM)} algorithm.

\noindent
\underline{\sc A sparsity measure}

We devise a sparsity measure that seeks to maximize the  traffic between two groups, compared  to the traffic
sent outside the merged groups.  This measure is very similar to the capacity on a fundamental cut explained in Section \ref{sec:rndoverview}.

In the setting of the capped hose model, the maximum possible demand between a pair of disjoint sets of nodes $A,B \subseteq V$ is given by
\begin{equation*}
\label{eqn:ustar}
\begin{array}{rl}
u^*(A,B) = \text{maximize} &   \sum_{i \in A, j \in B}  D_{ij} \\
\text{subject to}&D_{ij} \leq U_{ij} ~~~~\mbox{for all $i,j$}\\
&\sum_{j} D_{ij} \leq U(i) ~~~~\mbox{for all $i$}\\
& D_{ii} = 0 ~~~~\mbox{for all $i$}\\
&D \geq 0
\end{array}
.
\end{equation*}
This problem can  be solved efficiently, since it can be cast as a $b$-matching or max-flow problem. For instance,
consider a bipartite graph $H$ with bipartition $A,B$. Each edge $ij$ has a capacity of
$U_{ij}$ and is oriented from $A$ to $B$. We also add a source $s$ with edges
$(s,i)$ for $i \in A$ with capacity $U(i)$. Similarly we add sink $t$ with edges from $B$.
A max $st$-flow  in this graph gives precisely the value $u^*(A,B)$.

The \emph{sparsity} of a pair of disjoint sets of nodes $A,B \subseteq V$ is then:
\begin{equation*}
\label{eqn:sparsity}
sc(A,B) = \frac{u^*(A\cup B, V\setminus (A\cup B))}{u^*(A,B)}.
\end{equation*}
One notes that it trades off the flow out of $A \cup B$ versus the flow between $A,B$.

\noindent
\underline{\sc The Sparsest Merging Algorithm}

Our proposed heuristic algorithm (see Algorithm 1) builds up a binary hub tree $T$, starting from
a forest consisting of a set of singletons. At each step it has a forest and looks for
a  pair of rooted trees  that has the minimum sparsity value. Specifically, it computes the sparsity
 measure for  the leaf sets for such a pair of trees.
  The two subtrees are then merged, that is, a new root node is added to $T$ and becomes
the parent of the two previous rooted subtrees.
When the forest becomes a single tree, we stop. Clearly this can be done in polytime.

\begin{algorithm}
\caption{Binary Tree with Sparsest Merging algorithm.}
\label{algo:BTSM}
\begin{algorithmic}
\REQUIRE A set of nodes nodes $V$ with peak demands $U(i,j)$ for each pair of nodes $i,j\in V$ and marginals $U(i)$ for each node $i\in V$.
\ENSURE A binary demand tree $T$.
\STATE $T=(V,\emptyset)$
\STATE $S\leftarrow V(T)$
\WHILE{$|S|>1$}
\STATE Find $i,j\in S$ that has minimum sparse cut value.
\STATE Add a new node $u$ to $T$.
\STATE Connect $u$ to $i$ and $j$.
\STATE Add $u$ to $S$.
\STATE Remove $i$ and $j$ from $S$
\ENDWHILE
\RETURN $T$
\end{algorithmic}
\end{algorithm}

\noindent
\underline{\sc Restricting to Binary Trees.}

In this section, we show that there is always a binary hub tree that is optimal
 for $\rndhh$.

Let $T$ be any hub tree that has at least one node $u$ of degree at least four. Then it is possible to form a new hub tree $T'$ with $c_\hh(T') \leq c_\hh(T)$, and such that the sum of the degrees of nodes with degree greater than four is smaller in $T'$. Hence, repeating the process eventually reduces the number of degree at least four nodes to zero, and we end up with a binary tree.

Consider a  $u$ of $T$ with degree at least four. Replace $u$ by two nodes $u'$ and $u''$ such that $u$'s old parent is now connected to $u'$, and $u'$ is connected to $u$'s first child as well as to  $u''$. Then  $u''$ is connected to the remaining children of $u$ (see Figure 2). Note that the sum of degrees of nodes with degree at least four is decreased by one, since $u'$ has degree three  and $u''$ has degree $\deg(u)-1$.

\begin{figure}
\begin{center}
\begin{tikzpicture}
\node[] (v) at (0,0){};
\node[draw,circle,inner sep=1.5pt,label=left:$u$] (u) at (0,-1){};

\node[] (1) at (-1,-2){};
\node[] (2) at (-0.5,-2){};
\node[] (3) at (0,-2){};
\node[] (4) at (0.5,-2){};
\node[] (5) at (1,-2){};

\draw (v)--(u);
\draw (u)--(1);
\draw (u)--(2);
\draw (u)--(3);
\draw (u)--(4);
\draw (u)--(5);
\end{tikzpicture}
\begin{tikzpicture}
\node[] (v) at (0,0){};
\node[draw,circle,inner sep=1.5pt,label=left:$u'$] (u') at (0,-1){};
\node[draw,circle,inner sep=1.5pt,label=right:$u''$] (u'') at (0.5,-2){};

\node[] (1) at (-0.5,-2){};
\node[] (2) at (0,-3){};
\node[] (3) at (0.25,-3){};
\node[] (4) at (0.75,-3){};
\node[] (5) at (1,-3){};

\draw (v)--(u');
\draw (u')--(u'');
\draw (u')--(1);
\draw (u'')--(2);
\draw (u'')--(3);
\draw (u'')--(4);
\draw (u'')--(5);
\end{tikzpicture}
\end{center}
\label{fig:binaring}
\caption{Dividing internal nodes to get a binary tree.}
\end{figure}
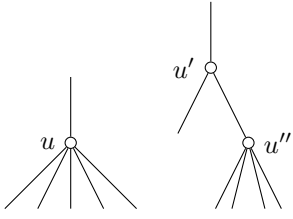

The capacity on the edge between $u$'s parent and $u'$ is equal to the capacity between $u$ and its parent;  the capacity between $u'$ and $u''$ can be inferred using equation \eqref{eqn:caps}. All other capacities remain the same - the capacity between one of $u$'s children and its new parent $u'$ or $u''$ is the same as the capacity between $u$ and that child in $T$. Note that we only split $u$ into two nodes and added enough capacity between those two nodes to route any feasible traffic matrix. Hence we have the new capacitated $T'$ can support exactly the same demand matrices as $T$ (i.e. the $T'$-tope and $T$-topes are the same).

Finally, an optimal hub map of $T$ can be implemented for  $T'$ as follows.
 One  maps $u'$ and $u''$ to the node that $u$ is mapped to. Hence, $T'$ can produce the same oblivious routing as $T$, so its optimal  cost $c_\hh(T')$ is at most $c_\hh(T)$.

\subsection{The Hub Placement Problem}
\label{sec:hubplacement}

Suppose we are given a network graph $G$, hub tree $T$ with leaf set $V(G)$ and edge capacities $u$.
In \cite{olver10approx},  a dynamic programming algorithm is given to find a valid mapping $\eta: V(T) \rightarrow V(G)$ which minimizes
the hub tree's capacity routing cost $c_\hh(T) = \sum_{e=wv \in T}  u(e) \dist_G(\eta(w),\eta(v))$, where $\dist_G(s,t)$ is the length of a shortest $st$-path in $G$ with edge costs $c$. We now describe how this is used for
general hub placement in $\rndhh(G,c,\mathcal{D})$.

Consider any demand between $i,j$. If $e=wv$ lies on the unique $ij$ path in $T$,
 then this demand ultimately routes on a shortest path between $\eta(w),\eta(v)$ in $G$.
Hence if there are demand matrices that route $X$ units of flow on $e$ in $T$, then we need to reserve $X$ units of capacity
on a (presumably shortest) $\eta(w)\eta(v)$ path in $G$. In other words, the $\rndhh$ hub placement  problem reduces
to the dynamic program mentioned above where the $u(e)$ values are  the solution to the following problem:
here, for $e \in T$,
$\mathsf{fund}(e)$ denotes all terminal pairs $i,j$ whose path in $T$ uses edge $e$ ($i$ and $j$ are separated by the fundamental cut induced by $e$).
If $\mathcal{D}$ is the demand polytope,
 then the maximum amount of demand across edge $e$ is
\begin{equation}
\label{eqn:caps}
u_{\mathcal{D},T}(e) := \max_{D \in \mathcal{D}} \sum_{i,j \in  \mathsf{fund}(e)}  D_{ij}.
\end{equation}
Thus a tree $T$ plus the polytope $\mathcal{D}$ induce some maximum capacities on the edges in $T$. These represent
how much capacity we need to install on each ``shortest path $\eta(w)\eta(v)$'' to route $\mathcal{D}$. These in turn
drive the dynamic program for placing the hubs.
 Note that we have chosen $u_{\mathcal{D},T}$ to be  the smallest capacity vector $u$ such that
$\mathcal{D} \subseteq \mathcal{H}_{T,u}$. In other words, designing a network to support the $T$-tope $\mathcal{H}_{T,u_{\mathcal{D},T}}$ also supports $\mathcal{D}$.

Under an arbitrary mapping $\eta:V(T) \rightarrow V(G)$, we must reserve $u_{\mathcal{D},T}(e=vw)$ units of capacity for the  circuit
on the shortest path between $\eta(v),\eta(w)$ in $G$. Hence the triple $T,\mathcal{D},\eta$ induces a total capacity cost of
\begin{equation}
\label{eqn:obj}
c_\hh(T) = \sum_{e=uv \in T}  u_{\mathcal{D},T}(e) \dist_G(\eta(v)\eta(w)).
\end{equation}

We refer to the {\em hub placement} problem as determining the map
$\eta$ which minimizes (\ref{eqn:obj}).
In \cite{olver10approx} a dynamic programming algorithm is given to find the optimal
mapping efficiently.  This forms the basis of their
$\O(1)$-approximation for the {\em Generalized VPN Problem}. Namely, the general $\rnd$ problem
for $T$-topes.

\subsection{$\rndhh$ and Other Algorithmic Problems}
\label{sec:cases}
$\rndhh$ contains well-studied algorithmic problems as special cases.
First, if $\mathcal{D}=\mathcal{H}_U$ is the class of hose matrices, then, as mentioned earlier, an optimal solution to $\rnd$ is
 induced by a hub routing template, which is in turn induced by a star hub tree. Hence in this case
$\rndhh = \rndhub = \rnd$.

Second, consider the case where $\mathcal{D}$ consists of a single demand matrix $D$.
Obviously,  the optimal
 solution is induced by a $\sp$ template. $\rndhh$ however, asks for
an optimal $\hh$ templates. This  is a
 classical problem in combinatorial optimization called the {\em minimum communication cost tree} (MCT).
 Find some tree (not necessarily in $G$), containing $V(G)$,
such that if we route the demands $D$ on $T$, the overall routing cost is minimized.\footnote{This
is directly related in turn to the average stretch tree problem \cite{peleg1998deterministic}.}
Being allowed to use a hub tree for the routing, as opposed to a subtree of $G$, essentially
means that we are restricting to MCT with metric costs\footnote{Costs are \emph{metric} if they are induced by a cost function that is a metric.}.
For the special case
 where  $G$ is itself a complete graph with unit-cost edges,  Hu \cite{hu1974optimum} showed that an optimal solution is induced by a so-called Gomory-Hu tree on the complete graph with edge capacities
 $D_{ij}$.\footnote{In general, the Gomory-Hu tree for a capacitated graph need not be subtree of the graph; hence
 the requirement that $G$ is complete.}

 As we now see, $\rndhh$ may incur a logarithmic increase over $\rnd$, and in fact $\rndsp$.
 On the positive side, it is known that in general there is a tree in the metric completion of
  $G$ whose ``distortion''  is at most $\O(\log n)$  \cite{fakcharoenphol2003tight}. The approach of Gupta
   thus yields an $\O(\log n)$ approximation for $\rndhh$.  One can see this by noting that $\rndhh$ costs no
   more than a best tree routing on $T$ in the metric completion. To see this, define
    a hub tree $T'$ by hanging a leaf edge $vv'$ from each node of $T$.  Then consider the mapping where
    each $v,v'$ maps to $v \in G$.
    This establishes:
\begin{fact}
\label{fact:hhvstr}
For any instance of $\rnd$,
$$\rndhh \leq \rndtr.$$
\end{fact}
    \noindent
    While these techniques  yield an $\O(\log n)$-approximation for  $\rndhh$
 and hence also MCT, it cannot be ruled out that there
are polytime $\O(1)$-approximations. For the purposes of our study, we use a heuristic algorithm instead of metric embeddings. This is outlined in Section \ref{sec:heuristic}.



\subsection{Basic Results for $\rndhh$}
\label{sec:gaps}

In this section we show that $\rndsp$ and $\rndhh$ are incomparable;
depending on the instance, one may be significantly cheaper than the
other. This motivates our main goal of classifying the demand spaces
according to which
template is better.

We now establish a family of $\rnd$ instances where $\rndhh$ does much worse than $\rnd$, and in fact worse than $\rndsp$.
\begin{theorem}
\label{thm:hhvssp}
There is a sequence of (unit-cost) graphs $\{G_n\}_{n\in\N}$ on $n$ nodes and demand polytopes $\{\mathcal{D}_n\}_{n\in\N}$ such that
for these instances
$$\rndhh (G_n) = \Omega(\log n \, \rndsp (G_n)).$$
\end{theorem}
\begin{proof}
We describe $G=G_n$ and $P=P_n$ for a fixed $n$.
The network graph $G$ consists of a $d$-regular, $\alpha$-expander graph on $n$ vertices, where $d=\mathcal{O}(1)$ and $\alpha=\Omega(1)$ (see \cite{hoory06expandergraphs} for existence proof). Directly we have that $|E(G)|=\O(n)$. Moreover,  all edges have unit cost.

The demand polytope $\mathcal{D}$ consists of a single matrix: one unit of demand between the endpoints of any   edge in $G$. The $\rndsp$ solution then consists of the network itself (route each demand on its own edge).
Its cost is $\O(n)$, the number of edges in the graph. Let $T$ be any hub tree. We show that the cost of any  $\hh$ solution induced by $T$ is $\Omega(n\log n)$. Hence, $\rndhh$  is $\Omega(n\log n)$.

Let $c\in V(T)$ be a \emph{center} of the tree, that is each component of $T\setminus\{c\}$ has at most $\frac{n}{2}$ leaves (it is well-known, and easy to see that every tree has a center). We call an edge $e\in E(G)$ \emph{separated by $c$ in $T$} if its endpoints lie in different components of $T\setminus \{c\}$. We let $E_c\subseteq E(G)$ be the set of all edges separated by $c$ in $T$. Let $C_i$ be one of the components in $T\setminus \{c\}$ and let $\ell(C_i)$ be the leaves of $C_i$. We know that $|\ell(C_i)|\leq \frac{n}{2}$. Hence, since $G$ is an $\alpha$-expander, there is at least $\alpha |\ell(C_i)|$ edges going out of $C_i$.  So we have that $|E_c|$ (or the total number of edges between the $C_i$'s) is at least
$$|E_c|\geq\frac{1}{2}\sum_{i} \alpha |\ell(C_i)| = \frac{\alpha}{2}\sum_{i} |\ell(C_i)| = \frac{\alpha n }{2}$$
as $\sum_{i} |\ell(C_i)| = n$ (the leaves of $T$ consist of the nodes of $G$).

Let $\eta:V(T)\to V(G)$ be the optimal hub map for $T$. Call a node $v\in V(G)$ \emph{close to $c$} if
$$\dist_G(\eta(v)=v,\eta(c))<\log_{d}(\min\{\frac{1}{2},\frac{\alpha}{2d}\} n).$$
Let $V_c$ be the set of close vertices in $G$. Since $G$ is a $d$-regular graph, its size is bounded by
$$|V_c|\leq d^{\log_{d}(\min\{\frac{1}{2},\frac{\alpha}{2d}\} n)} = \min\{\frac{1}{2},\frac{\alpha}{2d}\} n .$$
We now merge the two concepts and say that an edge $e=xy\in E_c$ that is separated by $c$ in $T$ is \emph{good} if both $x,y\in V_c$, i.e. $x,y$ are close to $c$. Conversely, an edge $e=xy\in E_c$ is \emph{bad} if one of its endpoints is not close to $c$ in $T$. Let $B$ be the set of bad edges. By the $d$-regularity of $G$, a bound on the number of bad edges is given by
$$|B|\geq |E_c|-\frac{d|V_c|}{2}\geq \frac{\alpha n}{2} - \frac{d\alpha n}{4d} \geq \frac{\alpha}{4} n.$$
Since the bad edges are unit demands that must be simultaneously routable in $T$, the ``non-close" endpoint of these demand edges must route through the image $\eta(c)$ of $c$, imposing a cost of at least $\log_d (n)$. We have $\Omega(n)$ bad demands, each incurring a cost of $\Omega(\log n)$, and thus the cost of the optimal \textsc{Hierarchical Hubbing} induced by $T$ routing is $\Omega(n\log n)$.
\end{proof}
The next result shows that $\rndtr$ may  be much cheaper than $\rndsp$. Hence, we can get an unbounded gap between $\rndsp$ and $\rndtr$. By Fact 1, this implies that $\rndhh$ may also be much cheaper than $\rndsp$.

\begin{theorem}
\label{thm:spvstrhh}
There is a sequence of weighted graphs $\{G_n\}_{n\in\N}$ and demand polytopes $\{\mathcal{D}_n\}_{n\in\N}$ such that
$$\rndsp (G_n) = \Omega(n^2)$$
and
$$\rndtr (G_n) = \O(1).$$
\end{theorem}
\begin{proof}
The network $G=(V,E)$ (actually $G_{2n+2}$)  consists of two stars centered at nodes $a,b$ connected by an edge.
The leaves are $v_i,v_i'$  with the following edges:
\begin{align*}
E = \{ av_i : 1 \leq i \leq n\} \cup \\
 \{ bv'_i : 1 \leq i \leq n\} \cup \\
  \{v_iv'_j : 1\leq i,j \leq n\} \cup \\
  \{ab\}.
\end{align*}

The edges are weighted as follows: the bridge edge $ab$  has weight one, the star edges $av_i$ or $bv'_i$ for $1\leq i \leq n$ have weight $\frac{1}{2n}$ and the edges $v_iv'_j$ for $1\leq i,j \leq n$ have weight $1-\frac{1}{n}$.

The (symmetric) demand polytope $\mathcal{D}$ consists of convex combinations of unit demands between $v_i$ and $v'_j$.
 That is, for each $i,j$ let $D^{i,j}$ be the matrix with
$D_{v_iv_j'}=D_{v_j'v_i}=1$ and all other entries $0$.
Then $\mathcal{D}=\conv(D^{i,j})$.

The optimal shortest path template in the network consists of routing the demand between $v_i$ and $v'_j$
 through the edge of cost $1-\frac{1}{n}$ connecting $v_i$ and $v'_j$. This template induces a cost of $\rndsp = n^2(1-\frac{1}{n})=n^2 - n = \Omega(n^2)$.

Now consider the tree template induces by the union of the two stars.
In order to support $\mathcal{D}$, the sub-optimization problem only  reserves one unit of capacity on the bridge edge $ab$, as well as one unit on all star edges. Hence, this capacity cost is $\O(1)$, and is an upper bound
for the best tree solution.
\end{proof}

\section{Evaluation Studies}
We compare the network design costs for capped hose models using the routing strategies $\sp$ and $\hh$. We solve the associated optimization problems ($\rndsp$, $\rndhh$) across many instances of the capped hose traffic model. For the purposes of this study, we implemented an exact method for $\sp$, and employed our BTSM heuristic in the case of $\hh$ routing.

For the evaluation set-up, we use two carrier network topologies: the American backbone network Abilene (11 nodes, 14 edges), and the Australian telecom network Telstra \cite{spring2002rocketfuel} (104 nodes, 151 edges).
We assume that per-unit link capacity costs are proportional to physical distances so we use this as our cost vector $c$ for determining the best $\rnd$ solution.
Our traffic data is based both on randomly generated traffic instances within the capped hose model constraints, as well as on a traffic scenario based on real data.

For the simulated traffic scenarios, for each network we randomly generate many traffic instances (each corresponding to a capped hose polytope) on which to solve $\rnd$. That is, each instance arises from some collection of peak demands and marginals. Our random instances are generated starting from actual population statistics (see Section~\ref{sec:generate}). In addition on Abilene, we use real data based on previous work from \cite{zhang2003tomo} which gives a time series of point-to-point demand measurements that can be used to create capped hose models (described later).

Our results show that $\hh$ templates are often the most cost-effective. That is, $\rndhh$ is very often less than both $\rndsp$, and $\rndhub$. In particular, this means that by adding peak capacities $U(i,j)$ into the hose model, the single-hub tree-routing template is no longer cost-effective. We  need additional hubs for a cost-effective network design. This is strong evidence for the use of $\hh$ routing.

\subsection{Generating Traffic}	
\label{sec:generate}

To generate our traffic demand matrices, namely marginals $U(i)$ and peak demands $U(i,j)$, we try to mimic how traffic is naturally generated in today’s core networks. It has been widely accepted that core traffic is stochastic in nature (``bursty'') and that peak demand  is much larger than average demand. This change in traffic patterns presents a moving target for service providers. The uncertainty is due to novel content-based network applications combined with factors such as data center consolidations and content mobility (content migrating from location to location based on where they are consumed).
Therefore in our evaluation scenario, we first assume peak demands, and then impose marginals based on the equipment choice at each node (capacity and number of ports at the ingress and egress nodes).
Our process is anchored by first fixing the peak demand between every pair of node locations based on their population \cite{chang2006many}. This gravity model takes into account the population size of the two sites as well as their distance. Since larger sites attract more traffic  and  closer proximity can  lead to  greater attraction, the gravity model incorporates these two features. The relative traffic between two places is determined by multiplying the population of one city by the population of the other city, and then dividing the product by the distance between the two cities squared. In our case, the peak demand between two cities $i,j$ is
\begin{equation}
\label{eq:peakdemands}
U(i,j) = \alpha(i,j)\,\mathsf{Population}(i)\,\mathsf{Population}(j)
\end{equation}
where $\alpha(i,j) = 1/\mathrm{distance}(i,j)$ and $\mathrm{distance}(i,j)$ is the geographical distance between node i and node j.

As datacenter based traffic dominates over the Internet, significant research activity is underway to learn the characteristics of this traffic. However, as noted in \cite{adhikari2010dynamics} we see that the end-to-end traffic matrix showing how the total traffic to YouTube video-servers are divided among different data centers follows the gravity model. On the other hand, the entry-exit traffic matrix seen by the ISP suggests that data-center traffic flows are not getting divided to different locations. This is due to the impact of BGP routing policies employed by the ISP to handle datacenter-bound traffic.

\noindent
\underline{\sc Generating marginals}

To generate marginals, we propose  sampling within the individual
 relevance intervals, as discussed in Section \ref{sec:cappedhose}. We start  by discretizing $U(i)$'s relevance interval into $s>0$ steps, for each of the $n$ nodes $i$. Then, for each value $\sigma_i \in \{0,1,2 \dots ,s\}$ we get a $U(i)$ within its relevance intervals. Formally:
$$
U(i) = \max_j U(i,j) + \frac{\sigma_i}{s} (\sum_j U(i,j) - \max_j U(i,j)).
$$

Thus each vector $\sigma \in\{0,1,\ldots,s\}^n$ yields an instantiation of the marginals.
For fixed $s$ and $n$, this gives  $(s+1)^n$ possible data points, yielding all possible marginals within their relevance interval up to a precision based on the coarseness of the discretization (size of $s$).

It is impractical to generate all $(s+1)^n$ points for large $n$ and $s$. To bypass this, we use a smaller value $k$ in place of $n$ to generate $(s+1)^k$ points. We then  sample uniformly from these $\sigma$ to assign the values for each node, i.e. we mimic the distribution of the short vector $\sigma$ of length $k$ to create our longer marginal vector of length $n$. So we get
\begin{equation}
\label{eq:marginal}
U(i) = \max_j U(i,j) + \frac{\sigma_{\mathcal{U}[k]}}{s} (\sum_j U(i,j) - \max_j U(i,j)).
\end{equation}
where $\sigma \in \{0,1,\ldots,s\}^k$ and $\mathcal{U}[k]$ is chosen uniformly at random from $\{1,2,\ldots,k\}$.

\subsection{Time-Series Traffic Matrix Data}
\label{sec:real}
For the Abilene network, we use the estimation of point-to-point demands from the work of \cite{zhang2003tomo} to create a realistic data point. The data consists of 24 weeks worth of traffic matrices sampled each five minutes. In other words, we get a sequence $\{D^{t}_{i,j}\}_{t=1}^{T}$ ($T=48384$) of demands from $i$ to $j$ for all nodes $i,j$ in the network. We transform this data into marginals and peak demands for a capped hose model as follows. The peak demand between two nodes $i,j$ is by definition $U(i,j) = \max_t D^{t}_{i,j}.$ Similarly, node $i$'s marginal is the (minimum) capacity required to route the maximum traffic out of that node at a given time step: $U(i) = \max_t \sum_j D^t_{i,j}.$ We consider different sampling intervals, i.e. various $T$'s.

\subsection{Computation of Link Costs}
It is simple to compute the capacity $u(e)$ needed on each edge $e$ for $\hh$ - we just add up the capacities from every hub tree edge $f=ij$ for which the $\eta(i)\eta(j)$-path passes through $e$. For $\sp$, we follow a similar sub-optimization as used for $\hh$ templates. Namely, we find the maximum value of $u(e)=\sum_{ij}w_{ij}(e) D_{ij}$ over all demand matrices $D$ in our capped hose universe $\mathcal{D}$ where $w_{ij}(e)$ is $1$ if $P_{ij}$ uses edge $e$, and is $0$ otherwise.
For any edge $e=st$, it is straightforward to see that the set of $ij$'s with $w_{ij}(e)=1$ induces a bipartite graph $H=(V(G),\{ij:w_{ij}=1\})$ when the $P_{ij}$'s are shortest paths. Suppose $H$ is not bipartite, then there is an odd cycle $C=i_1i_2\ldots i_{2k+1}$. Thus, $w_{i_{2k+1}i_1}=1$ and $w_{i_ji_{j+1}}=1$ for all $1\leq j \leq 2k+1$. Without loss of generality, the fact that $P_{ij}$'s are shortest paths implies that $i_j$ is closer to $s$ for odd $j$'s and closer $t$ for even $j$'s (see Figure 3). But then this means the path $P'_{i_{2k+1}i_1} = \sp_G(i_{2k+1},s) \cup \sp_G(s,i_1)$ is strictly shorter than $P_{i_{2k+1},i_1}$ as it doesn't use $e$; contradiction.
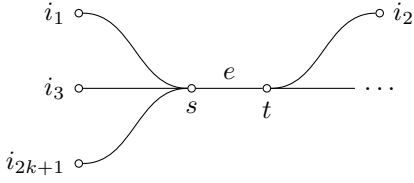
\begin{figure}
\centering
\begin{tikzpicture}
\node[inner sep=1pt,draw,circle,label=left:$i_1$](i1) at (0,0){};
\node[inner sep=1pt,draw,circle,label=right:$i_2$](i2) at (4,0){};
\node[inner sep=1pt,draw,circle,label=left:$i_3$](i3) at (0,-1){};
\node[](i4) at (4,-1){$\ldots$};
\node[inner sep=1pt,draw,circle,label=left:$i_{2k+1}$](i5) at (0,-2){};

\node[inner sep=1pt,draw,circle,label=below:$s$](s) at (1.5,-1){};
\node[inner sep=1pt,draw,circle,label=below:$t$](t) at (2.5,-1){};
\draw (s)--(t) node[pos=0.5,above]{$e$};

\draw (i1) to[out=0,in=180] (s);
\draw (i3) to[out=0,in=180] (s);
\draw (i5) to[out=0,in=180] (s);
\draw (i2) to[out=180,in=0] (t);
\draw (i4) to[out=180,in=0] (t);
\end{tikzpicture}
\caption{Going from $i_j$ to $i_{j+1}$, always through $e=st$ and along shortest paths.}
\label{fig:bipartiteargument}
\end{figure}
 Hence we may use a max flow routing algorithm to compute $u(e)$.
The computation of link capacities can thus be done with a collection of all shortest paths $\{P_{ij}\}$ generated using standard algorithms.

Algorithm 2 illustrates our experimental procedure.
\begin{algorithm}[!h]
\label{alg:exp}
\caption{Experimental procedure.}
\begin{algorithmic}
\REQUIRE Network graph $G=(V,E)$ (where edge costs correspond to the geographical distance between edge's endpoints) and parameters $s,k\in \N$.
\STATE
\FOR{\textbf{each} pair of node cities $i,j\in V$}
\STATE Fix point-to-point peak demand between $i$ and $j$ as per Equation \eqref{eq:peakdemands}.
\ENDFOR
\FOR{\textbf{each} $\sigma\in (s+1)^k$}
\FOR{\textbf{each} $i\in V$}
\STATE Fix node $i$'s marginal as per Equation \eqref{eq:marginal}.
\ENDFOR
\STATE This gives a capped hose model $c\mathcal{H}$ arising from marginals $U(i)$'s and peak demands $U(i,j)$'s.
\STATE Find marginal strength vector $\mu$ of $c\mathcal{H}$.
\STATE Find peak demand strength vector $\pi$ of $c\mathcal{H}$.
\STATE Solve $\rnd_\sp(G,c\mathcal{H})$ optimally.
\STATE Use our BTSM heuristic and Olver et al's algorithm \cite{olver10approx} to approximately solve $\rnd_\hh(G,c\mathcal{H})$ and get a solution of cost $\rnd'_\hh(G,c\mathcal{H})$.
\STATE Report $\|\mu\|_2$, $\|\pi\|_2$ and $\frac{\rnd_\sp(G,c\mathcal{H})}{\rnd'_\hh(G,c\mathcal{H})}$.
\ENDFOR
\end{algorithmic}
\end{algorithm}

\subsection{Results}
In Figures 2 and 3, we plot the results of the low-cost routing templates for two different network topologies Abilene and Telstra\footnote{A much smaller number of points were sampled to cover most of the spectrum due to the computational constraints caused by the size of the Telstra network.}. Traffic matrix instances are plotted according to their marginal and peak demand strengths and coloured based on the ratio between the edge cost of the (optimal) $\sp$ routing template and the $\hh$ routing template. A blue point indicates that $\sp$ is better and a red point indicates that $\hh$ is better. Moreover, the intensity of the colouring is proportional to the proximity of the ratio to one. If the costs of $\sp$ and $\hh$ are very similar, the color is almost white, otherwise it is highly saturated

Intuitively, and as discussed in Section \ref{sec:strength}, the $x$-axis corresponds to the affinity of the traffic to the hose model. A large $\mu$ value means more constraining marginals. Thus the optimal routing should primarily consider the marginals only. So we expect hub routing to do much better. Similarly, the $y$-axis corresponds to the affinity of the traffic to the peak-demand model. Again, a large $\pi$ value means more constraining peak demands, and thus we expect $\sp$ to do much better. Furthermore, there is a duality between the $\mu$ and $\pi$ values: a demand universe can only have strong marginals or strong peak demands, not both. This explains the quarter circle ``swoosh'' shape of the plot.

\begin{figure}
\centering
\includegraphics[scale=0.45]{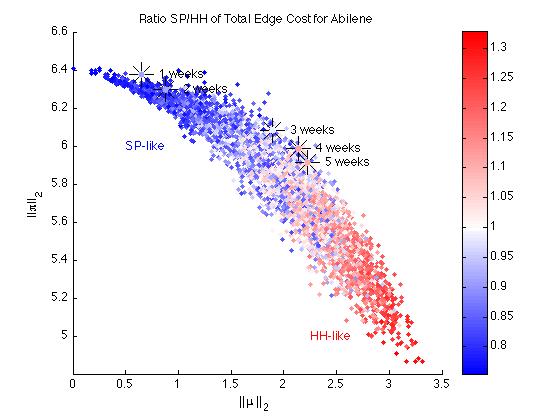}
\caption{The ratio between the edge cost of the (optimal) $\sp$ routing and the $\hh$ routing found with our heuristic plotted against the norms of the peak demand and marginal strength vectors for varying traffic, and the Abilene network. Data points from the time-series of traffic matrices added, with duration considered as label. \emph{Better viewed in color.}}
\end{figure}
\begin{figure}
\centering
\label{fig:TELSTRAplot}
\includegraphics[scale=0.45]{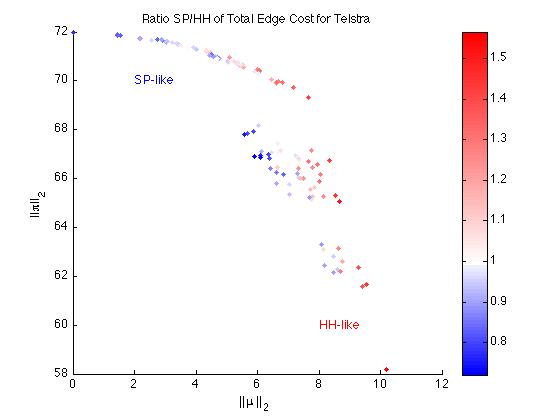}
\caption{The ratio between the edge cost of the (optimal) $\sp$ routing and the $\hh$ routing found with our heuristic plotted against the norms of the peak demand and marginal strength vectors for varying traffic, and the Telstra network. \emph{Better viewed in color.}
}
\end{figure}

\section{Discussion}
In \cite{shepherd2006selective}, hub routing was used to compare the cost of two architectures: one based on IP routing
and one employing circuits. Our work
 focuses on how the link and/or node costs vary with  the choice of routing template.
We now address some of the issues and our findings.

\subsection{$\mu-\pi$ as indicators for the choice of a routing template}
We see that there is a continuum of routing templates from $\hh$ to $\sp$ if viewed
as a function of the newly introduced $\mu$ and $\pi$ measures. As such we propose that it is possible to define a routing indicator that accounts for the strengths of the marginals and peak demands and use this information to choose the appropriate routing template. This is an important observation because it suggests a new strategy for service providers faced with designing for changing demands. In addition to just using marginal demands, they can sample the demand space to determine the relative strengths of the peak and marginal demands. They may
 use this to  choose an initial routing template and if necessary, evolve the  template to obtain/maintain
  cost-efficient routing. This further justifies the need to measure link level traffic, both peak and average demand, since it can feedback into the selection of an optimal routing template  for a given network topology.

An important consideration is the sampling interval. The impact of the sampling interval can be clearly seen by the evolution of the real data points in Figures 4 and 5. Each such data point comes from a compilation of the real traffic matrices over a certain period of time, as labelled next to the points. As the interval increases from one week to four, we see the capped hose model shift toward a more hose-like universe, with stronger marginal (and consequently weaker peak demands).
If samples are over a short time horizon, then the demand matrix does not change significantly and as such we are designing  for a single demand. In this scenario,  $\sp$ yields the cheapest routing template. However as the time horizon gets longer, the different demand matrices could be significantly different and in such scenarios the more flexible design of $\hh$/$\vpn$ brings in the added cost benefits.

\subsection{Single versus Multi Hub Routing Templates}

An interesting observation of this study is that there are significant cost benefits for using a multi-hub routing template versus the single-hub templates which yield optimal $\vpn$s. Table \ref{tab:hubcount} shows the number of traffic instances where $\hh$ is more cost effective than a single hub $\vpn$ solution.
It is straightforward to see that a best solution induced by any hub tree is always  as good (in terms of total cost) as a single-hub routing. Indeed, mapping all internal nodes of a hub tree to a single network node is always a valid hub placement that yields the single-hub $\vpn$-like routing. Since the hub placement algorithm \cite{olver10approx} is optimal, it is guaranteed to find a hub placement that is as good as the best single-hub routing.

\begin{table}
\centering
\caption{Number of Instances where $\hh$ Outperforms $\sp$}
\begin{tabular}{c||ccc}
Network & $\hh$ Optimal & $\hh$ Optimal\\
				 & Multiple Hubs & Single Hub \\
\hline
Abilene  {(\scriptsize total number of instances: 6561)}  &    3045           &			8		  \\
Telstra {(\scriptsize total number of instances: 353)} &        246               &			0		
\end{tabular}
\label{tab:hubcount}
\end{table}

\subsection{Routing Costs and Impact on Network Design}
We evaluate the cost-effectiveness of the $\hh$ template against $\sp$ using both link and node costs. Link or edge capacity costs for $\hh$ are computed by solving a maximization problem over the universe of traffic matrices. For node costs we use the maximum possible link capacities that are incident on a particular node. The intuition behind using these maximum incident capacities is based on the hardware port requirements for routing/switching the total traffic through that node. However it is very rare for the all incident links at the node to simultaneously carry their maximum allowable traffic rates.

From Figures 4 and 6 for link and node costs respectively, we see that $\hh$ shows significant cost savings for link costs since we take advantage of dense traffic sharing points through appropriate location of hub nodes. However, since the node costs are being evaluated using physical link capacities the gains are not as significant.

\begin{figure}
\centering
\includegraphics[scale=0.45]{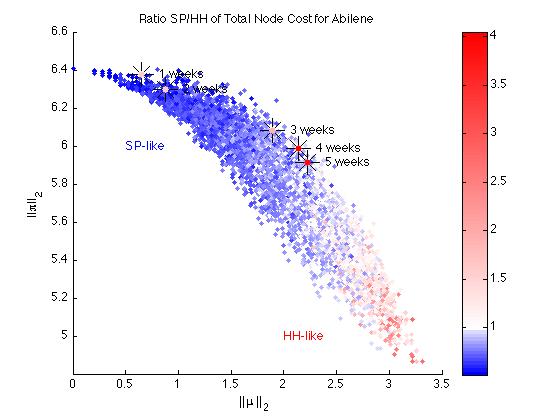}
\caption{The ratio between the node cost of the (optimal) $\sp$ routing and the $\hh$ routing found with our heuristic plotted against the norms of the peak demand and marginal strength vectors for varying traffic, and the Abilene network. Data points from the time-series of traffic matrices added, with duration considered as label. \emph{Better viewed in color.}}
\end{figure}

\section{Conclusion}
The capped hose model introduced in this paper shines light on the applicability of current routing template archetypes ($\hub$ and $\sp$) parametrized by the marginal and peak demands of the given traffic matrices. We have shown that using the capped hose model provides additional flexibility to network engineers for cost-optimal network design and shows the relevance of the newly introduced $\hh$ templates. In particular, this means that by adding peak capacities $U(i,j)$ into the hose model, the single-hub tree routing template is no longer cost-effective. Moreover, the peak demand and marginal strengths presented here are a first attempt at choosing routing templates based solely on traffic information. Finally, the $\rndhh$ problem is of theoretical interest by itself, encompassing many other well-known algorithmic problems, and is open to further investigation.

\bibliographystyle{IEEEtran}
\bibliography{paperbib}

\end{document}